\documentclass[runningheads,a4paper]{llncs}

\usepackage{amsmath, amssymb}
\usepackage{graphicx}
\usepackage{times} 
\usepackage{subfigure}
\usepackage{sidecap}
\usepackage{multirow} 
\usepackage{wrapfig}
\usepackage{capt-of}

\usepackage{pgfplots}
\pgfplotsset{compat=1.3}
\usepgfplotslibrary{groupplots} 
\usepgfplotslibrary[groupplots] 
\usetikzlibrary{pgfplots.groupplots} 
\usetikzlibrary[pgfplots.groupplots] 
\usetikzlibrary{patterns}

\usepackage{url}
\urldef{\mailsa}\path|{alfred.hofmann, ursula.barth, ingrid.haas, frank.holzwarth,|
\urldef{\mailsb}\path|anna.kramer, leonie.kunz, christine.reiss, nicole.sator,|
\urldef{\mailsc}\path|erika.siebert-cole, peter.strasser, lncs}@springer.com|

\DeclareMathOperator{\LCA}{LCA}

\begin{document}

\mainmatter  

\title{Inferring Species Trees from Incongruent Multi-Copy Gene Trees Using the Robinson-Foulds Distance}

\titlerunning{Inferring Species Trees from Gene Trees Using the Robinson-Foulds Distance}

%
%
\author{Ruchi Chaudhary\inst{1} \and J. Gordon Burleigh\inst{2} \and David  Fern\'{a}ndez-Baca\inst{1}}
%

\institute{Department of Computer Science, Iowa State University, Ames, IA 50011, USA
\and Department of Biology, University of Florida, Gainesville, FL 32611, USA}

%
%

\maketitle

\begin{abstract}
We present a new method for inferring species trees from multi-copy gene trees. Our method is based on a generalization of the Robinson-Foulds (RF) distance to multi-labeled trees (mul-trees), i.e., gene trees in which multiple leaves can have the same label. Unlike most previous phylogenetic methods using gene trees, this method does not assume that gene tree incongruence is caused by a single, specific biological process, such as gene duplication and loss, deep coalescence, or lateral gene transfer.  We prove that it is NP-hard to compute the RF distance between two mul-trees, but it is easy to calculate the generalized RF distance between a mul-tree and a singly-labeled tree. Motivated by this observation, we formulate the RF supertree problem for mul-trees (MulRF), which takes a collection of mul-trees and constructs a species tree that minimizes the total RF distance from the input mul-trees. We present a fast heuristic algorithm for the MulRF supertree problem.  Simulation experiments demonstrate that the MulRF method produces more accurate species trees than gene tree parsimony methods when incongruence is caused by gene tree error, duplications and losses, and/or lateral gene transfer.  Furthermore, the MulRF heuristic runs quickly on data sets containing hundreds of trees with up to a hundred taxa.
\end{abstract}

\section{Introduction}
With the development and spread of next generation sequencing technologies, there is great interest in incorporating large genomic data sets into phylogenetic inference. One challenge for such phylogenomic analyses is that genes sampled from the same set of species often produce conflicting trees \cite{maddison97}. Some of the incongruence may be due to errors in the phylogenetic analyses \cite{Swofford:1996:PI}.  The discordance also may reflect evolutionary events such as recombination, gene duplication, gene loss, deep coalescence, and lateral gene transfer (LGT) \cite{Avise:1983:MBE,Doyle:1992:SB,Goodman:1979:FTG,Maddison:1996,maddison97,Pamilo:1988:MBE}. Indeed, under certain conditions the most likely gene tree topology to evolve along a species tree will differ from the species tree~\cite{DegnanRosenberg2006}. Thus, in order to construct phylogenetic hypotheses from genomic data, it is necessary to address the incongruence among gene trees.

Approaches to inferring species from conflicting gene trees typically use a model of gene evolution that can reconcile the gene tree and species tree topologies. In practice, these models are usually based on a single evolutionary mechanism, such as duplication and loss or deep coalescence. Although these models greatly simplify the true processes of genome evolution, more complex and realistic models can quickly become unwieldy, making it hard or impossible to analyze large genomic data sets. In this paper, we take a step back and approach the question of finding a species tree for a given collection of gene trees though a method that is based on a tree distance metric and does not imply any specific evolutionary mechanism.

\paragraph{Previous Work.}
Existing methods for inferring species trees from collections of gene trees can be divided into two broad categories: non-parametric methods based on gene tree parsimony (GTP), and likelihood-based approaches \cite{Ane07012007,Kubatko04012009,Liu06012007}. GTP methods take a collection of discordant gene trees and try to find the species tree that implies the fewest evolutionary events. GeneTree \cite{Page:1998:GCG}, DupTree \cite{Wehe:bioinfo:2008}, and DupLoss \cite{Bansal:APBC:2010} seek to minimize the number of duplications or duplications and losses.  GeneTree \cite{Page:1998:GCG}, Mesquite \cite{maddison97}, PhyloNet \cite{Yu:recomb:11}, and the method of \cite{Bansal:APBC:2010} minimize deep coalescence events. The Subtree Prune and Regraft (SPR) supertree method \cite{Whidden:SPR:2012} is based on minimizing the number of LGT events.  Some of these methods are quite fast, enabling the analysis of very large data sets, but errors in the gene trees can mislead GTP analyses \cite{Burleigh:SB:2011,Huang:2009:WDA,Sanderson2007}. Also, in some cases GTP methods may be statistically inconsistent \cite{Than:2011:MDC}. Many of the likelihood-based methods use coalescence models to reconcile gene tree topologies \cite{Kubatko04012009,Liu06012007}. Although such likelihood-based approaches have a firm statistical basis, they often are computationally expensive.

While all the existing methods differ widely in their details, at a high level, except \cite{Ane07012007}, they all are based on potentially restrictive assumptions about  the source of discordance among gene trees.

\paragraph{Our Contributions.} We present a species tree inference technique that is not linked to any specific mechanism of gene tree discordance and has the scalability and accuracy expected for genome-wide analyses for many taxa.  Our method takes as input a collection of multi-labeled gene trees (mul-trees), trees where multiple leaves can have the same label, and finds a species tree at minimum ``distance'' to the input trees. The ability to use mul-trees as input, instead of being restricted to single copy genes, allows this method to incorporate the wealth of genomic data from multi-copy genes into phylogenetic inference, not only single-copy genes. Our distance measure is a generalization of the Robinson-Foulds (RF) distance to mul-trees. The RF distance has been useful as a supertree method for singly-labeled input trees \cite{Mukul:2010:RFS,Ruchi:2012:URF}, and in the singly-labeled setting, the distance based approach may be statistically consistent \cite{SteelRodrigo08}.

Our contributions are as follows:

\begin{itemize}
\item We study the problem of computing the RF distance between two mul-trees, and show that it is NP-hard (Section \ref{sec:prelim}).
  \item We formulate a RF supertree problem for mul-trees, which we call MulRF, that takes a collection of mul-trees as input and constructs a supertree that is at minimum RF distance from each input mul-tree (Section \ref{sec:model}).  A key component of this approach is a simple and efficient technique to compute the RF distance between an input mul-tree and a \emph{singly-labeled} species tree.  (Note the contrast with the previously-mentioned NP-hardness result.)
  \item We provide a fast heuristic algorithm for the MulRF problem (Section \ref{sec:solution}).  Heuristics are needed for this problem because it is NP-hard. 
  \item We implemented the MulRF algorithm and performed experiments on complex gene tree simulations (Section \ref{sec:experiment}).
  \end{itemize}

Simulation experiments allow us to evaluate the accuracy of our method by comparing it against the true species tree, something that cannot be done on real data. We compared the supertrees constructed by MulRF and GTP methods that consider only duplication \cite{Wehe:bioinfo:2008}, duplication and loss \cite{Bansal:APBC:2010}, and only LGT \cite{Whidden:SPR:2012} with the true species trees.  Likelihood-based methods were not considered because the simulated gene trees were comparatively large in size for these methods and no likelihood-based phylogenetic method deals explicitly with duplication and loss or LGT . In all experiments, MulRF produced trees that are more similar to the true species trees than those obtained by other three methods. Further, our algorithm ran quickly on moderate-size data sets, finishing in under two minutes on  data sets containing 300 gene trees evolved over 100 taxon species trees, suggesting it is scalable for large-scale phylogenomic analyses.


\section{Preliminaries}
\label{sec:prelim}
A \emph{phylogenetic tree} or \emph{tree} is an unrooted, leaf-labeled tree in which all the internal vertices have degree of at least three \cite{Semple:2003:phy}. The leaf set of $T$ is denoted by $\mathcal{L}(T)$. The set of all vertices of $T$ is denoted by $V(T)$ and the set of all edges by $E(T)$. The set of all internal vertices of $T$ is $I(T) := V(T) \backslash \mathcal{L}(T)$.
A tree is \emph{binary} if every internal vertex has degree three. Let $U$ be a subset of $V(T)$. We denote by $T[U]$ the minimum subtree of $T$ that connects the elements in $U$. The \emph{restriction} of $T$ to $U$, denoted by $T_{|U}$, is the phylogenetic tree that is obtained from $T[U]$ by suppressing all vertices of degree two.

Two trees $T_1$ and $T_2$ are isomorphic if there exists a bijection $\tau : V(T_1) \rightarrow V(T_2)$ such that $\{u,v\} \in E(T_1)$ if and only if  $\{\tau(u),\tau(v)\} \in E(T_2)$ for all $\{u,v\} \in {V(T_1) \choose 2}$.

The \emph{contraction} of an edge in a tree collapses that edge and identifies its two endpoints. The \emph{refinement} of an unresolved vertex (i.e., an internal vertex with
degree greater than three) expands that  vertex into two vertices connected by an edge.  Contraction and refinement can be viewed as inverses of each other (Fig. \ref{conref}).

\begin{figure}[t]
\vspace*{-0.05in}
 \centering
 \includegraphics[width=3in]{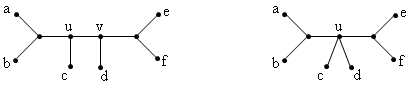}
 \vspace*{-0.18in}
 \caption{The contraction of edge $\{u,v\}$ in the first tree produces the second tree; conversely, the refinement of vertex $u$ in the second tree produces the first tree.}
 \label{conref}
 \vspace*{-0.2in}
 \end{figure}

The \emph{Robinson-Foulds (RF)} distance between two trees $T_1$ and $T_2$, denoted by $RF(T_1,$ $T_2)$, is the minimum number of contractions and refinements necessary to transform $T_1$ into a tree isomorphic to $T_2$ \cite{Robinson:1981:CPT}. The RF distance between two trees can be equivalently defined via \emph{splits}. A split $A|B$ is a bipartition of the leaf set of a tree; $A$ and $B$ are the \emph{parts} of split $A|B$. The set of all splits induced by the internal edges of a tree $T$ is denoted by $\Sigma(T)$. Now for $T_1$ and $T_2$ \cite{Robinson:1981:CPT}, $$RF(T_1, T_2) := |(\Sigma(T_1) \backslash \Sigma(T_2)) \cup (\Sigma(T_2) \backslash \Sigma(T_1)) |.$$
Two trees $T_1$ and $T_2$ are isomorphic if $\Sigma(T_1) = \Sigma(T_2)$ \cite[page 44]{Semple:2003:phy}.

A \emph{phylogenetic mul-tree} or \emph{mul-tree}, is a tuple $\mathcal{T} = (T,M,\varphi)$ consisting of an unrooted tree $T$, a set of labels $M$, and a surjective \emph{labeling function} $\varphi: \mathcal{L}(T) \rightarrow M$ that maps each leaf of $T$ with a label in $M$. Informally, a mul-tree is simply a phylogeny in which multiple leaves can have the same label (see Fig. \ref{counter}). For any label $\ell \in M$, $\varphi^{-1}(\ell)$ is the set of all leaves labeled $\ell$. If $\varphi$ is a bijection, the corresponding mul-tree is just a (singly-labeled) tree. In this paper, we use the traditional notation for a tree when the given mul-tree is clearly a tree.

 The concepts introduced above for unrooted trees naturally extend to mul-trees. For example, a mul-tree $\mathcal{T} = (T,M,\varphi)$ is binary if $T$ is binary. Two mul-trees $\mathcal{T}_1 = (T_1,M,\varphi_1)$ and $\mathcal{T}_2 = (T_2,M,\varphi_2)$ are isomorphic if $T_1$ and $T_2$ are isomorphic under bijection $\tau : V(T_1) \rightarrow V(T_2)$ such that $\varphi_1(u) = \varphi_2(\tau(u))$ for all $u \in \mathcal{L}(T_1)$.

\begin{figure}
\vspace*{-0.3in}
 \centering
 \includegraphics[width=3in]{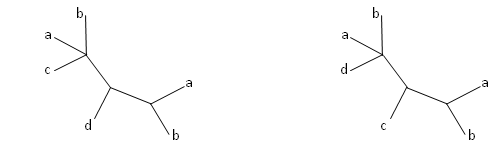}
 \vspace*{-0.2in}
 \caption{Two mul-trees that induce the same set of splits but are not isomorphic.}
 \label{counter}
 \vspace*{-0.25in}
 \end{figure}

The contraction and refinement based RF distance metric naturally extends to mul-trees \cite{Ganapathy:2006:PIB}. However, unlike singly-labeled trees, it is possible for two mul-trees $\mathcal{T}_1$ and $\mathcal{T}_2$ to satisfy $\Sigma(\mathcal{T}_1) = \Sigma(\mathcal{T}_2)$ and yet not be isomorphic (see Fig. \ref{counter}). Thus, the RF distance between two mul-trees cannot be computed by splits. Ganapathy et al.~gave a worst-case exponential time algorithm for computing the RF distance between two mul-trees \cite{Ganapathy:2006:PIB}.  The next result suggests that a polynomial time algorithm is unlikely.

\begin{theorem}\label{thm:NP} Computing the RF distance between two mul-trees is NP-hard.
\footnote{The proofs of this and other results are in the Appendix.}
\end{theorem}

\section{MulRF Supertrees}

\label{sec:model}

A \emph{profile} is a tuple of mul-trees $\mathcal{P} := (\mathcal{T}_1, \mathcal{T}_2,...,\mathcal{T}_k)$, also called \emph{input trees}, where  $\mathcal{T}_i = (T_i,M_i,\varphi_i)$ for each $i \in \{1,\dots, k\}$. A \emph{supertree} on $\mathcal{P}$ is a singly-labeled phylogenetic tree $S$ such that $\mathcal{L}(S) = \bigcup _{i = 1}^{k} M_i$. We write $n$ to denote $|\mathcal{L}(S)|$, the total number of distinct leaves in the profile. In this paper, we assume that the size of each input mul-tree differs only by a constant factor from the size of the resulting supertree.

We extend the notion of RF distance to the case where $\mathcal{L}(T_1) \subseteq \mathcal{L}(T_2)$ by letting $RF(T_1, T_2) := RF(T_1, {T_2}_{|\mathcal{L}(T_1)})$.
We define the \emph{RF distance} from a profile $\mathcal{P}$ to a supertree $S$ for $\mathcal{P}$ as $RF(\mathcal{P},S) := \sum _{\mathcal{T} \in \mathcal{P}} RF(\mathcal{T},S)$.

Let $\mathcal{B}(\mathcal{P})$ be the set of all binary supertrees for $\mathcal P$.

\begin{problem} [RF Supertree for MUL-Trees (MulRF)]  \\
\textit{Input:} A profile $\mathcal{P} = (\mathcal{T}_1, \mathcal{T}_2,...,\mathcal{T}_k)$ of unrooted mul-trees. \\
\textit{Output:} A supertree $S$* for $\mathcal{P}$ such that $RF(\mathcal{P},S\text{*}) =  \min_{S \in \mathcal{B}(\mathcal{P})} RF(\mathcal{P},S)$.  \end{problem}

The MulRF problem is NP-hard even when all the input mul-trees are singly-labeled trees on the same leaf set \cite{McMorris:Steel:93}.  In fact, as stated in Theorem~\ref{thm:NP}, just computing the RF distance between two mul-trees is hard.  Nevertheless, we now show that it is straightforward to compute the RF distance between an input mul-tree and a supertree.

Let $\mathcal{T} = (T,M,\varphi)$ be an input mul-tree and $S$ be a supertree, where $M \subseteq \mathcal{L}(S)$. The \emph{extended supertree} is the mul-tree $\mathcal{S}$ constructed from $S$ by replacing each $a \in \mathcal{L}(S)$ by an internal node connecting to $k$ leaves labeled with $a$, where $k := |\varphi^{-1}(a)| > 1$. See Fig. \ref{fig:ext}.
A \emph{full differentiation} of $\mathcal{T}$ is a leaf labeled tree $\mathbf{T}$ such that $T$ and $\mathbf{T}$ are isomorphic.

Let $\mathcal{T} = (T,M,\varphi)$ and $\mathcal{S} = (T',M',\varphi')$ be two unrooted mul-trees.  Two  full differentiations $\mathbf{T}$ and $\mathbf{S}$ of $\mathcal{T}$ and $\mathcal{S}$, respectively, are \emph{consistent} if  for each $a \in M \cap M'$, $\tau_1(\varphi^{-1}(a))=\tau_2(\varphi'^{-1}(a))$, where $T$ and $\mathbf{T}$ are isomorphic under bijection
$\tau_1 : V(T) \rightarrow V(\mathbf{T})$ and $T'$ and $\mathbf{S}$ are isomorphic under bijection $\tau_2 : V(T') \rightarrow V(\mathbf{S})$. For instance,
a consistent full differentiation can be obtained by relabeling each of the $k$ copies of each leaf label $a$ by $a_1, a_2, \dots ,a_k$ in both the trees.

%



\begin{figure}
\vspace*{-0.25in}
\begin{minipage}[b]{0.3\linewidth}
\centering
$\mathcal{T}$ \end{minipage}
\begin{minipage}[b]{0.3\linewidth}
\centering
$S$\end{minipage}
\begin{minipage}[b]{0.4\linewidth}
\centering
$\mathcal{S}$\end{minipage}

\begin{minipage}[b]{0.3\linewidth}
\centering
\raisebox{0.2\height}{\includegraphics*[width=0.7\textwidth]{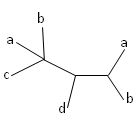}}
\end{minipage}
\begin{minipage}[b]{0.3\linewidth}
\centering
\raisebox{0.3\height}{\includegraphics*[width=0.7\textwidth]{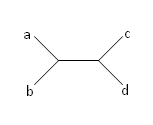}}
\end{minipage}
\raisebox{15\height}{$\Longrightarrow$}
\begin{minipage}[b]{0.3\linewidth}
\centering
\raisebox{0.3\height}{\includegraphics*[width=0.6\textwidth]{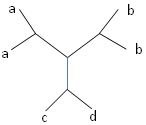}}
\end{minipage}
\label{fig:ext}
\vspace*{-0.4in}
\caption{Input mul-trees $\mathcal{T}$ and the supertree $S$. The extended supertree $\mathcal{S}$ is also shown.}
\vspace*{-0.2in}
\end{figure}

\begin{theorem}[\cite{Ganapathy:2006:PIB}] Let $\mathcal{T}$ and $\mathcal{S}$ be two mul-trees. Then, $RF(\mathcal{T},\mathcal{S}) = \min\{RF(\mathbf{T},\mathbf{S}): \mathbf{T}$ and $\mathbf{S}$ are mutually consistent full differentiations of $\mathcal{T}$ and $\mathcal{S}$, respectively$\}$.
\end{theorem}

\begin{theorem} \label{thm:diffr} Let $\mathcal{T}$ be an input mul-tree and $\mathcal{S}$ be the extended supertree. Then, all mutually consistent full differentiations of $\mathcal{T}$ and $\mathcal{S}$ give the same RF distance. \end{theorem}

In short, the RF distance between an input mul-tree and a supertree can be computed by 1) extending the supertree, 2) producing one consistent full differentiation of the two mul-trees, and 3) applying the split based formula to compute the RF distance.

\section{Solving the MulRF Problem}
\label{sec:solution}
Our local search heuristic for the MulRF problem starts with an initial supertree and explores the space of possible supertrees in search of a \emph{locally optimum} supertree; i.e., a tree whose score is minimum within its ``neighborhood''. The neighborhood is defined in terms of the \emph{Subtree Prune and Regraft} (\emph{SPR}) operation \cite{Allen:2001:STO}.  An SPR operation on an unrooted, binary tree $T$ cuts any edge, thereby pruning a subtree $t$, and then regrafts $t$ by the same cut edge to a new vertex obtained by subdividing a pre-existing edge in $T-t$ (Fig.~\ref{spr}).
The set of all trees obtained by the application of a single SPR operation on $T$ is called the \emph{SPR neighborhood} of $T$, and is denoted by $SPR_T$.  The size of this neighborhood is $\Theta(n^2)$.

\begin{figure}
\vspace*{-0.23in}
 \centering
 \includegraphics[width=2.9in]{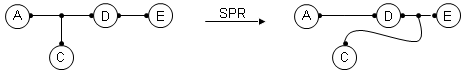}
 \vspace*{-0.15in}
 \caption{A schematic representation of the SPR operation.}
 \label{spr}
 \vspace*{-0.3in}
 \end{figure}

\begin{problem} [SPR Search] \\
\textit{Input:} A profile $\mathcal{P} = (\mathcal{T}_1, \mathcal{T}_2,...,\mathcal{T}_k)$ of unrooted mul-trees and a binary supertree $S$ for $\mathcal{P}$.\\
\textit{Output:} A tree $S\text{*} \in SPR_S$ such that $RF(\mathcal{P},S\text{*}) = \min_{S' \in SPR_S} RF(\mathcal{P},S')$. \end{problem}

In Section~\ref{sec:SPRSearch}, we present  an algorithm for the SPR search problem that runs in time $\Theta(n^2k)$.  The algorithm relies on results from \cite{Ruchi:2012:URF}, which characterize the RF distance between unrooted trees in terms of least common ancestors in rooted versions of those trees.
These properties enable us to update the RF distance quickly after an SPR operation has been applied to one of the trees.
For completeness, we briefly review these results in the next subsection.  For a full discussion with proofs, see \cite{Ruchi:2012:URF}.

\subsection{Robinson-Foulds Distance and Least Common Ancestors}
\label{structural}

A \emph{rooted phylogenetic tree}  $\mathbb{T}$ has exactly one distinguished vertex  $rt(\mathbb{T})$, called the \emph{root}. The root is a degree-two vertex if the tree is binary.
A vertex $v$  of $\mathbb{T}$ is \emph{internal} if $v \in  V(\mathbb{T})\backslash(\mathcal{L}(\mathbb{T}) \cup rt(\mathbb{T}))$.  The set of all internal vertices of $\mathbb{T}$ is denoted by $I(\mathbb{T})$. We define $\preceq_{\mathbb{T}}$ to be the partial order on $V(\mathbb{T})$ where $x \preceq_{\mathbb{T}} y$ if $y$ is a vertex on the path from $rt(\mathbb{T})$ to $x$. If $\{x, y\} \in E(\mathbb{T})$ and $x \preceq_{\mathbb{T}} y$, then $y$ is the \emph{parent} of $x$ and $x$ is a \emph{child} of $y$.
The \emph{least common ancestor (LCA)} of a non-empty subset $L \subseteq V(\mathbb{T})$, denoted by $\LCA_{\mathbb{T}}(L)$, is the unique smallest upper bound of $L$ under $\preceq_{\mathbb{T}}$.

Let $\mathbb{T}_v$ denote the subtree of $\mathbb{T}$ rooted at vertex $v \in V(\mathbb{T})$. For each node $v \in I(\mathbb{T})$, $C_{\mathbb{T}}(v)$ is defined to be the set of all leaf nodes in $\mathbb{T}_v$. Set $C_{\mathbb{T}}(v)$ is called a \emph{cluster}. Let  $\mathcal{H}(\mathbb{T})$ denote the set of all clusters of $\mathbb{T}$. The RF distance between rooted trees $\mathbb{T}$, $\mathbb{S}$ over the same leaf set
is defined as  \cite{Robinson:1981:CPT}$$RF(\mathbb{T},\mathbb{S}) := |(\mathcal{H}(\mathbb{T}) \backslash \mathcal{H}(\mathbb{S})) \cup (\mathcal{H}(\mathbb{S}) \backslash \mathcal{H}(\mathbb{T}))|.$$

Let $\mathbb{T}$ and $\mathbb{S}$ be the trees that result from rooting $\mathbf{T}$ and $\mathbf{S}$ at the branches incident on some arbitrarily-chosen but fixed taxon $r \in \mathcal{L}(\mathbf{T}) \cap \mathcal{L}(\mathbf{S})$  (Fig. \ref{root}).

\begin{figure}
\vspace*{-0.25in}
   \centering
   \includegraphics[width=3in,angle=0]{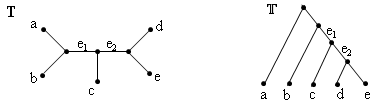}
\vspace*{-0.2in}
\caption{Tree $\mathbf{T}$ with leaf set $\{a, b, c, d, e\}$. The rooted tree $\mathbb{T}$ with $r=a$ is also shown.}
\label{root}
\vspace*{-0.32in}
\end{figure}

\begin{lemma}[\cite{Ruchi:2012:URF}] \label{lm:rooted-RF}
Let $\mathbf{T}$ and $\mathbf{S}$ be two unrooted phylogenetic trees with $\mathcal{L}(\mathbf{T}) = \mathcal{L}(\mathbf{S})$, then $RF(\mathbf{T},\mathbf{S}) = RF(\mathbb{T},\mathbb{S}).$ \end{lemma}

We extend RF distance to the case where $\mathcal{L}(\mathbb{T}) \subseteq \mathcal{L}(\mathbb{S})$ in the same way as for unrooted trees. That is, $RF(\mathbb{T}, \mathbb{S}) := RF(\mathbb{T}, \mathbb{S}_{|\mathcal{L}(\mathbb{T})})$, where $\mathbb{S}_{|\mathcal{L}(\mathbb{T})}$ is the rooted phylogenetic tree obtained from $\mathbb{S}[\mathcal{L}(\mathbb{T})]$ by suppressing all non-root degree-two vertices.

We now show how to compute the RF distance in this general setting, without explicitly building $\mathbb{S}_{|\mathcal{L}(\mathbb{T})}$.  We need two concepts.  Let $v \in I(\mathbb{S})$. The \emph{restriction} of $C_{\mathbb{S}}(v)$ to $\mathcal{L}(\mathbb{T})$ is $\hat{C}_{\mathbb{T}}(v) := \{w \in \mathcal{L}(\mathbb{S}_v): w \in \mathcal{L}(\mathbb{T})\}.$
The \emph{vertex function} $f_{\mathbb{S}}$ assigns each $u \in I(\mathbb{T})$ the value $f_{\mathbb{S}}(u) = |U|$, where $U := \{v \in I(\mathbb{S}): C_{\mathbb{T}}(u) = \hat{C}_{\mathbb{T}}(v)\}$.  Observe that if $\mathcal{L}(\mathbb{S}) = \mathcal{L}(\mathbb{T})$, then for all $u \in I(\mathbb{T})$, $f_{\mathbb{S}}(u) \leq 1$.




\begin{lemma}[\cite{Ruchi:2012:URF}] \label{RFComp} $RF(\mathbb{T},\mathbb{S}) = |\mathcal{L}(\mathbb{T})|-|I(\mathbb{T})|+2|\mathcal{F}_{\mathbb{S}}|-2$, where $\mathcal{F_{\mathbb{S}}} := \{u \in I(\mathbb{T}): f_{\mathbb{S}}(u) = 0 \}$. \end{lemma}

\begin{figure}
\vspace*{-0.35in}
   \centering
   \includegraphics[width=2.5in]{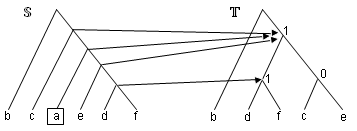}
\vspace*{-0.2in}
\caption{The LCA mapping from $\mathbb{S}$ to $\mathbb{T}$. Vertex $a$ in $\mathbb{S}$ is mapped to \emph{null} as $a \notin \mathcal{L}(\mathbb{T})$. The  internal vertices of $\mathbb{T}$ are labeled with the values of the vertex function.}
\label{map}
\vspace*{-0.2in}
\end{figure}

We now describe a $O(n)$-time algorithm to compute the initial vertex function for $\mathbb{S}$ relative to $\mathbb{T}$, along with the RF distance between these two trees.  The algorithm relies on LCAs.
For $\mathbb{S}$ and $\mathbb{T}$, the \emph{LCA mapping} $\mathcal{M}_{\mathbb{S},\mathbb{T}}: V(\mathbb{S}) \rightarrow V(\mathbb{T})$ is defined as $$\mathcal{M}_{\mathbb{S},\mathbb{T}}(u) :=
\begin{cases}
\text{LCA}_{\mathbb{T}}(\hat{C}_{\mathbb{T}}(u)), & \text{if } \hat{C}_{\mathbb{T}}(u) \neq \phi \text{ ;}\\
 \text{null}, & \text{otherwise.}
\end{cases}$$
See Fig. \ref{map}.

\begin{lemma}[\cite{Ruchi:2012:URF}] \label{LCAlemma} For all $u \in I(\mathbb{T})$, $f_{\mathbb{S}}(u) = |B|$, where $B := \{v \in I(\mathbb{S}): \mathcal{M}_{\mathbb{S},\mathbb{T}}(v) = u$ and $|C_{\mathbb{T}}(u)| = |\hat{C}_{\mathbb{T}}(v)|\}$.
\end{lemma}


The LCA computation for $\mathbb{T}$ can be done in $O(n)$ time, and the LCA mapping from $\mathbb{S}$ to $\mathbb{T}$ can be done in $O(n)$ time \cite{Bender:2000:TLP} in bottom-up manner. Further, from Lemmas \ref{RFComp} and \ref{LCAlemma} we can compute the RF distance between $\mathbb{S}$ and $\mathbb{T}$ in $O(n)$ time as well.  

\subsection{Solving the SPR Search Problem}

\label{sec:SPRSearch}

Let $\mathcal{T} = (T,M,\varphi)$ be an arbitrary mul-tree in  $\mathcal{P}$.  We now show how to compute the RF distance from $\mathcal{T}$ to each tree in the SPR$_S$ neighborhood in linear time of the size of the neighborhood. Let $\mathcal{S}$ be the supertree $S$ after extending for $\mathcal{T}$. Let $\mathbf{T}$ and $\mathbf{S}$ be any two  mutually consistent full differentiations  of $\mathcal{T}$ and $\mathcal{S}$, respectively.
By Theorem~\ref{thm:diffr}, computing the RF distance between an input mul-tree $\mathcal{T}$ and all trees in the SPR neighborhood of an extended supertree $\mathcal{S}$ reduces to finding the RF distance between $\mathbf{T}$ and each tree in the SPR neighborhood of $\mathbf{S}$.


Suppose an SPR operation on $S$ cuts the edge $e = \{x,y\}$, and that $X$, $Y$ are the subtrees of $S-e$ containing $x$, $y$, respectively. Suppose subtree $Y$ is pruned and regrafted by the same cut edge to a new vertex obtained by subdividing an edge in $X$. The degree-two vertex $x$ is suppressed and the new vertex is denoted by $x$. Observe that there are $O(n)$ possible edges in $X$ to regraft $Y$. We perform regrafts in an order that leads to a constant time RF distance computation for each successive regraft. \\

\noindent \textbf{Observation 1.}  \emph{For $Z \in \{X,Y\}$, if $M \cap \mathcal{L}(Z) = \emptyset$, then RF$(S',\mathcal{T}) =$ RF$(S,\mathcal{T})$ for each $S'$ obtained from $S$ by regrafting $Y$ on any edge in $X$.} \\


We begin by regrafting $Y$ at an edge incident to a leaf in $X$. Let $\overline{S}$ and $\overline{\mathbf{S}}$ denote, respectively,  the tree that results from performing the prune-and-regraft and the full differentiation of this result tree. 
We compute the RF distance between $\mathbf{T}$ and $\overline{\mathbf{S}}$ using the algorithm described in the previous section. This method works by computing the RF distance between the rooted trees $\mathbb{T}$ and $\overline{\mathbb{S}}$ obtained by rooting $\mathbf{T}$ and $\overline{\mathbf{S}}$ at any leaf labeled by an element of $M \cap \mathcal{L}(X)$. (Note that,  by Observation 1, if $M \cap \mathcal{L}(X) = \emptyset$, then $\mathcal{T}$'s distance from $\overline{S}$ is same as $S$.)  The algorithm also computes the LCAs for $\mathbb{T}$ and the LCA mapping from $\overline{\mathbb{S}}$ to $\mathbb{T}$.

We perform the remaining regrafts of $Y$ on edges in $X$ by iterating through the vertices of $X$, starting from a leaf and exploring as far as possible along each branch before backtracking. The $k^{th}$ regraft is performed on the edge between the $k^{th}$ and $k+1^{st}$ vertices in this iteration. Let us denote this ordering of edges by $\aleph$. See Fig.~\ref{iteration}.
Observe that each two distinct consecutive edges in $\aleph$ are adjacent. We will show that, after the initial RF distance computation for $\overline{S}$, we can compute in constant time the RF distance for the result of regrafting on each successive (adjacent) edges in $\aleph$.


\begin{figure}
\vspace*{-0.22in}
  \begin{minipage}{.5\linewidth}
    \centering
    \includegraphics[width=1.2in]{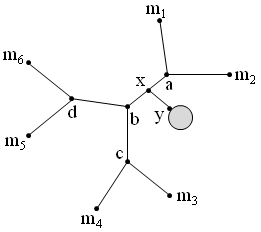}
  \end{minipage}
  \begin{minipage}{.45\linewidth}
    \captionof{figure}{A tree with a subtree regrafted at edge $\{a,b\}$. One iteration of vertices in the tree is $m_1, a, m_2,$ $a, b, c,m_3, c, m_4,c, b, d,m_5,$ $d, m_6,d, b, a,m_1$. The resulting ordering $\aleph$ is $\{m_1,a\},$ $\{a,m_2\},...,\{a,m_1\}$.}
    \label{iteration}
  \end{minipage}
  \vspace*{-0.22in}
\end{figure}

Beginning with $\overline{S}$, each $S' \in SPR_S$ helps in computing the RF distance of the next tree in the above regraft order. Assume that $S' \in SPR_S$ results from regrafting $Y$ at edge $\{a,b\}$ in $X$ as shown in Fig.~\ref{iteration}. Let the rooted tree obtained after extending and differentiating $S'$ be denoted by $\mathbb{S}'$. The LCA mapping and RF distance have been computed for $\mathbb{S}'$. Let $S'' \in SPR_S$ denote the tree obtained by regrafting $Y$ on edge $\{b,c\}$ in $X$ and the rooted counterpart of $S''$ is $\mathbb{S}''$.

Next, we find the vertices of $\mathbb{S}''$ whose LCA mapping  $\mathcal{M}_{\mathbb{S}'',\mathbb{T}}$ has changed as a result of the SPR operation. Based on the topology of $\mathbb{S}'$, there are three cases:
\begin{enumerate}
  \item \emph{$x$ is parent of $b$ and $b$ is parent of $c$.}  For all $t \in I(\mathbb{S}'') \backslash \{x,b\}$, $\mathcal{M}_{\mathbb{S}'',\mathbb{T}}(t)$ = $\mathcal{M}_{\mathbb{S}',\mathbb{T}}(t)$. Further, $\mathcal{M}_{\mathbb{S}'',\mathbb{T}}(b) := \mathcal{M}_{\mathbb{S}',\mathbb{T}}(x)$, and $\mathcal{M}_{\mathbb{S}'',\mathbb{T}}(x) := \LCA(\mathcal{M}_{\mathbb{S}',\mathbb{T}}(c),$ $\mathcal{M}_{\mathbb{S}',\mathbb{T}}(y))$.
  \item \emph{$b$ is parent of $c$ and $x$.} For all $t \in I(\mathbb{S}'') \backslash \{x\}$, $\mathcal{M}_{\mathbb{S}'',\mathbb{T}}(t)$ = $\mathcal{M}_{\mathbb{S}',\mathbb{T}}(t)$. Further, $\mathcal{M}_{\mathbb{S}'',\mathbb{T}}(x) := \LCA(\mathcal{M}_{\mathbb{S}',\mathbb{T}}(c),\mathcal{M}_{\mathbb{S}',\mathbb{T}}(y))$.
  \item \emph{$b$ is parent of $x$ and $c$ is parent of $b$.} For all $t \in I(\mathbb{S}'') \backslash \{b,x\}$, $\mathcal{M}_{\mathbb{S}'',\mathbb{T}}(t)$ = $\mathcal{M}_{\mathbb{S}',\mathbb{T}}(t)$. Moreover, $\mathcal{M}_{\mathbb{S}'',\mathbb{T}}(x) := \mathcal{M}_{\mathbb{S}',\mathbb{T}}(b)$, and $\mathcal{M}_{\mathbb{S}'',\mathbb{T}}(b) := \LCA(\mathcal{M}_{\mathbb{S}',\mathbb{T}}(d),$ $\mathcal{M}_{\mathbb{S}',\mathbb{T}}(a))$.
\end{enumerate}

Since we can check in constant time which one of the above three cases holds, the LCA mappings can be updated in constant time too. Let $H$ be a set $\{u \in I(\mathbb{T}) : f_{\mathbb{S}''}(u) \neq f_{\mathbb{S}'}(u)\}$. Set $H$ can be computed in constant time. Observe that $H$ has at most four vertices. Let $G$ denotes the set $\{w \in H: f_{\mathbb{S}'}(w) = 0, \text{ but } f_{\mathbb{S}''}(w) \geq 1\}$, and $L$ denote the set $\{w \in H: f_{\mathbb{S}'}(w) \geq 1, \text{ but } f_{\mathbb{S}''}(w) = 0\}$.

\begin{lemma}  \label{lm:NNI-RF} $RF(\mathbb{S}'',\mathbb{T}) = RF(\mathbb{S}',\mathbb{T}) - 2|G| + 2|L|$. \end{lemma}

Thus, after the initial regraft of $Y$ at a leaf in $X$, we can compute in constant time the RF-distance between $\mathbf{T}$ and the supertree that results from each subsequent regraft.


\begin{lemma} \label{lm:comp} For each $\{x,y\} \in E(S)$, where $X$ and $Y$ are two resulting subtrees containing $x$ and $y$, respectively. The RF distance for the set of trees obtained by regrafting $X$ (resp. $Y$) on each edge in $Y$ (resp. $X$) can be computed in $\Theta(n)$ time. \end{lemma}

\begin{theorem} \label{tm:main} The SPR Search problem can be solved in $\Theta(n^2k)$ time. \end{theorem}

\section{Experimental Evaluation}
\label{sec:experiment}
\subsection{Method}
\textbf{\emph{Simulated data sets.}} We generated model species trees using the uniform speciation (Yule) module in the program Mesquite \cite{Maddison:Mesquite:2009}. Two sets of model trees were generated: i) 50 taxa trees of height 220 thousand years (tyrs), ii) 100 taxa trees of height 440 tyrs (note that the dates are relative; they do not have to represent thousands of years). Each data set had 20 model species trees. We evolved 150 and 300 gene trees for each 50- and 100-taxon model species tree, respectively. 
We used Arvestad et al.'s \cite{ArvestadBLS03} duplication-loss model  to evolve gene trees within the model tree. 
We applied LGT events on the evolved gene trees, using the standard subtree transfer model of LGT.  One LGT event causes the subtree rooted at a vertex $c$ to be pruned and regrafted at an edge $(a,b)$, where $a$ and $b$ together are not in the path from the root (of the tree) to $c$. We used gene duplication and loss (D/L) rate of 0.002 events/gene per tyrs and LGT rate of 2 events per gene tree. In other words, a gene tree can have 0 to 2 LGT events.

We evolved gene trees based on four evolutionary scenarios: i) no duplications, losses, or LGT (called \emph{none}), ii) D/L rate 0.002 and no LGT (called \emph{dl}), iii) no duplication or loss, and LGT rate 2 (called \emph{lgt}), and iv) D/L rate 0.002 and LGT rate 2 (called \emph{both}). The parameter values for each simulation are called the \emph{model condition}. We deleted 0 to 25\% of the taxa (selected at random) from each gene tree  to represent missing data, which is common in almost all phylogenomic studies.. For each gene tree, we used Seq-Gen \cite{Rambaut:1997:SGA} to simulate a DNA sequence alignment of length 500 based on the GTR+Gamma+I model. The parameters of the model were chosen with equal probability from the parameter sets estimated in \cite{Ganapathy_PhDThesis06} on three biological data sets \cite{SwensonBWL2010}. We estimated maximum likelihood trees from each simulated sequence alignment using RAxML \cite{StamatakisBioinf2006}, performing searches from 5 different starting trees and saving the best tree. We rooted each estimated gene tree at the midpoint of the longest leaf-to-leaf path before the species tree construction. \\

\noindent\textbf{\emph{Species tree estimation.}} We estimated species trees via GTP minimizing only the number of duplications (Only-dup) \cite{Wehe:bioinfo:2008}, GTP minimizing duplications and losses (Dup-loss) \cite{Bansal:APBC:2010}, GTP minimizing  LGT events (SPR supertree or SPRS for short) \cite{Whidden:SPR:2012}, and the MulRF heuristic. Both Only-dup and Dup-loss were executed with their default settings, including a fast leaf-adding heuristic for initial species tree construction. SPRS was run with 25 iterations of the global rearrangement search option. For 50-taxon data sets, it calculated the exact rSPR distance if it was 15 or less, and otherwise it estimated the rSPR distance using the 3-approximation. For the 100-taxon data sets, we used the 3-approximation of the rSPR distance. SPRS does not allow mul-trees as input.  Therefore we only ran it on \emph{none} and \emph{lgt} data sets. Experiments were performed on the University of Florida High Performance Computing test nodes with 8 to 24 cores. \\ 

\begin{table}
\footnotesize
\vspace{-0.25in}
 \centering
\begin{tabular}{|c|c|c|c|c|c|}
\hline
Num. Taxa & Sets & Only-dup & Dup-loss & SPRS & MulRF \\
\hline\hline
\multirow{4}{*}{50} & \emph{none} & $<1$s & 2s & 8h 34m 32s & 3s    \\
                    & \emph{lgt}  & $<1$s & 2s & 8h 30m 30s & 2s       \\
                    & \emph{dl}   & $<1$s & 3s & NA & 6s        \\
                    & \emph{both} & $<1$s & 3s & NA & 6s     \\ \hline
\multirow{4}{*}{100} & \emph{none} & 9s & 37s & 21h 34m 25s & 58s     \\
                    & \emph{lgt}  & 11s & 49s & 19h 6m 9s & 51s        \\
                    & \emph{dl}   &  9s & 30s & NA & 1m 11s        \\
                    & \emph{both} & 11s & 37s & NA & 1m 15s     \\ \hline
\end{tabular}
\vspace{4pt}
\caption{Running time for species tree estimations}
\vspace{-0.35in}
\label{timeTable}
\end{table}

\noindent \textbf{\emph{Performance evaluation.}} We report the average topological error (ATE) for each model condition. This is the average of the normalized RF distance (dividing the RF distance by number of internal edges in both trees) between each of the 20 model species trees and their estimated species trees. An ATE of 0 indicates that two trees are identical, and an ATE of 100 indicates that two trees share no common splits. We also compared the number of gene duplications estimated by Only-dup and Dup-loss and losses estimated by Dup-loss with the actual number of these events in each gene tree simulation.

\subsection{Results}
Both Dup-loss and Only-dup overestimate duplications for sets \emph{dl} and \emph{both} in both 50- and 100-taxon model trees (Fig. \ref{Lossplot}(a,b)). They also imply many duplications in the \emph{none} and \emph{lgt} data sets, where the simulations included no duplications. Similarly, Dup-loss overestimates losses for sets \emph{dl} and \emph{both} and also erroneously estimates losses for sets \emph{none} and \emph{lgt} (Fig. \ref{Lossplot}(c,d)).

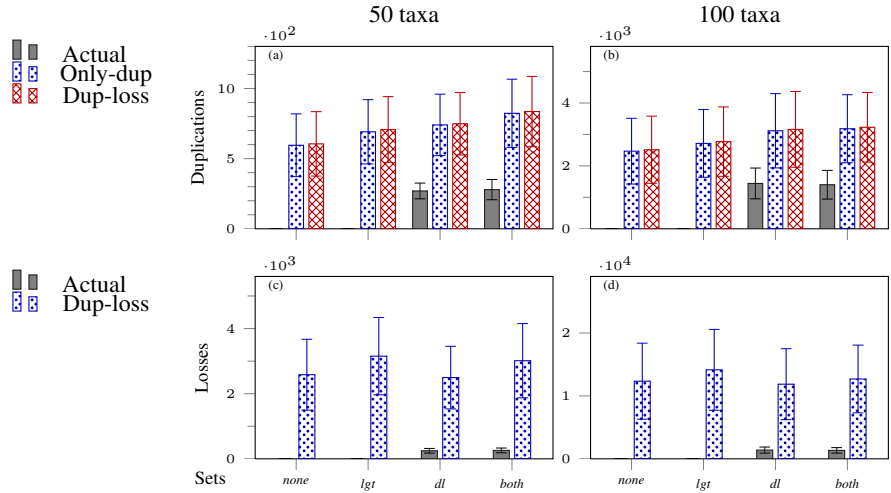
\begin{figure}
\vspace*{-0.2in}
\centering
\begin{tikzpicture}
    \matrix {
        \begin{axis}[
            width=5.5cm,height=4cm,xtick={1,2,3,4},xticklabels={,,,,},tick pos=left,ytick align=outside,xtick align=outside,minor y tick num=4,
            ylabel=Duplications,enlargelimits=0.19,scaled y ticks=base 10:-2, y label style={at={(-0.13,0.5)}}, x label style={at={(-0.15,-0.02)}},label style={font=\scriptsize},
            ybar,bar width=5.5pt,tick label style={font=\tiny},
            ymin=0.0,ymax=1300.0,title=50 taxa,enlarge y limits=false,
            extra description/.code={\node at (-0.8,0.93) {\ref{pgfplots:label1}}; \node at (-0.54,0.96) {Actual}; \node at (-0.8,0.81) {\ref{pgfplots:label2}}; \node at (-0.5,0.84) {Only-dup}; \node at (-0.8,0.69) {\ref{pgfplots:label3}}; \node at (-0.5,0.72) {Dup-loss}; \node at (0.07,0.948) {\tiny{(a)}};
            }
        ]
        \addplot[color=black!90!white,fill=black!50!white][error bars/.cd,y dir=both,y explicit,error mark options={rotate=90,black,mark size=2pt,line width=0.3pt}]
        table[x=rt,y=actual,y error=eractual] {Data/taxa-50D.txt}; \label{pgfplots:label1}
        \addplot[pattern=crosshatch dots,pattern color=blue!70!black,draw=blue!70!black,samples=700][error bars/.cd,y dir=both,y explicit,error mark options={rotate=90,blue!70!black,mark size=2pt,line width=0.3pt}]
        table[x=rt,y=od,y error=erod] {Data/taxa-50D.txt}; \label{pgfplots:label2}
        \addplot[pattern=crosshatch,pattern color=red!70!black,draw=red!70!black][error bars/.cd,y dir=both,y explicit,error mark options={rotate=90,red!70!black,mark size=2pt,line width=0.3pt}]
        table[x=rt,y=dl,y error=erdl] {Data/taxa-50D.txt};  \label{pgfplots:label3}
        \end{axis}
        &
        \begin{axis}[
            width=5.5cm,height=4cm,xtick={1,2,3,4},xticklabels={,,,,},tick pos=left,ytick align=outside,xtick align=outside,minor y tick num=1,
            enlargelimits=0.19,scaled y ticks=base 10:-3,
            ybar,bar width=5.5pt,tick label style={font=\tiny},
            ymin=0.0,ymax=5800.0,title=100 taxa,enlarge y limits=false,
            extra description/.code={\node at (0.07,0.948) {\tiny{(b)}};}
        ]
        \addplot[black!90!white,fill=black!50!white][error bars/.cd,y dir=both,y explicit,error mark options={rotate=90,black,mark size=2pt,line width=0.3pt}] table[x=rt,y=actual,y error=eractual] {Data/taxa-100D.txt};
        \addplot[pattern=crosshatch dots,pattern color=blue!70!black,draw=blue!70!black,samples=700][error bars/.cd,y dir=both,y explicit,error mark options={rotate=90,blue!70!black,mark size=2pt,line width=0.3pt}] table[x=rt,y=od,y error=erod] {Data/taxa-100D.txt};
        \addplot[pattern=crosshatch,pattern color=red!70!black,draw=red!70!black][error bars/.cd,y dir=both,y explicit,error mark options={rotate=90,red!70!black,mark size=2pt,line width=0.3pt}] table[x=rt,y=dl,y error=erdl] {Data/taxa-100D.txt};
        \end{axis}
        \\
        \begin{axis}[
            width=5.5cm,height=4cm,xtick={1,2,3,4},xticklabels={\emph{none},\emph{lgt},\emph{dl},\emph{both}},tick pos=left,ytick align=outside,xtick align=outside,minor y tick num=1,
            ylabel=Losses,xlabel={Sets},enlargelimits=0.19,scaled y ticks=base 10:-3, y label style={at={(-0.13,0.5)}}, x label style={at={(-0.15,-0.02)}},label style={font=\scriptsize},
            ybar,bar width=6pt,tick label style={font=\tiny},
            ymin=0.0,ymax=5600.0,enlarge y limits=false,
            extra description/.code={\node at (-0.8,0.93) {\ref{pgfplots:label8}}; \node at (-0.54,0.96) {Actual}; \node at (-0.8,0.81) {\ref{pgfplots:label9}}; \node at (-0.5,0.84) {Dup-loss};
            \node at (0.07,0.948) {\tiny{(c)}};
            }
        ]
        \addplot[black!90!white,fill=black!50!white][error bars/.cd,y dir=both,y explicit,error mark options={rotate=90,black,mark size=2pt,line width=0.3pt}]
        table[x=rt,y=actual,y error=eractual] {Data/taxa-50L.txt}; \label{pgfplots:label8}
        \addplot[pattern=crosshatch dots,pattern color=blue!70!black,draw=blue!70!black,samples=700][error bars/.cd,y dir=both,y explicit,error mark options={rotate=90,blue!70!black,mark size=2pt,line width=0.3pt}]
        table[x=rt,y=dl,y error=erdl] {Data/taxa-50L.txt}; \label{pgfplots:label9}
        \end{axis}
        &
        \begin{axis}[
            width=5.5cm,height=4cm,xtick={1,2,3,4},xticklabels={\emph{none},\emph{lgt},\emph{dl},\emph{both}},tick pos=left,ytick align=outside,xtick align=outside,
            enlargelimits=0.19,scaled y ticks=base 10:-4,
            ybar,bar width=6pt,tick label style={font=\tiny},
            ymin=0.0,ymax=29000.0,enlarge y limits=false,
            extra description/.code={\node at (0.07,0.948) {\tiny{(d)}};}
        ]
        \addplot[black!90!white,fill=black!50!white][error bars/.cd,y dir=both,y explicit,error mark options={rotate=90,black,mark size=2pt,line width=0.3pt}] table[x=rt,y=actual,y error=eractual] {Data/taxa-100L.txt};
        \addplot[pattern=crosshatch dots,pattern color=blue!70!black,draw=blue!70!black,samples=700][error bars/.cd,y dir=both,y explicit,error mark options={rotate=90,blue!70!black,mark size=2pt,line width=0.3pt}] table[x=rt,y=dl,y error=erdl] {Data/taxa-100L.txt};
        \end{axis}
        \\
        };
\end{tikzpicture}
\vspace*{-0.2in}
\caption{Graphs a-b shows duplications estimated by Only-dup and Dup-loss, and Graphs c-d losses estimated by Dup-loss, against the actual number of these events in gene trees, for all model conditions; means and standard errors are shown.}
\vspace*{-0.25in}
\label{Lossplot}
\end{figure}

For each set of 50- and 100-taxon model trees, the MulRF species trees are more accurate than those produced by the other three methods.  For example, the ATE rate of MulRF is 16.75\% to 39.91\% lower than the method of lowest ATE rate among other three methods (Fig. \ref{fig:ATEplot}).

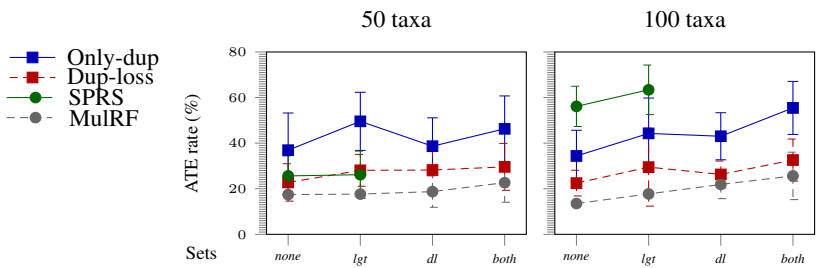
\begin{figure}
\vspace*{-0.3in}
\centering
\begin{tikzpicture}
    \matrix {
        \begin{axis}[
            width=5cm,height=4cm,xtick={1,2,3,4},xticklabels={\emph{none},\emph{lgt},\emph{dl},\emph{both}},tick pos=left,ytick align=outside,xtick align=outside,minor y tick num=19,
            ylabel={ATE rate (\%)},xlabel={Sets},x label style={at={(-0.25,-0.02)}}, y label style={at={(-0.22,0.5)}}, label style={font=\scriptsize},
            tick label style={font=\tiny},
            ymin=0,ymax=80,title=50 taxa,enlarge y limits=false,
            extra description/.code={\node at (-1.0,0.97) {\ref{pgfplots:label4}}; \node at (-0.59,0.96) {Only-dup}; \node at (-1.0,0.86) {\ref{pgfplots:label5}}; \node at (-0.6,0.855) {Dup-loss}; \node at (-1.0,0.75) {\ref{pgfplots:label6}}; \node at (-0.654,0.755) {SPRS}; \node at (-1.0,0.64) {\ref{pgfplots:label7}}; \node at (-0.626,0.64) {MulRF};
            }
        ]
        \addplot[color=blue!70!black,mark=square*][error bars/.cd,y dir=both,y explicit, error mark options={rotate=90,blue!70!black,mark size=2pt,line width=0.3pt}] table[x=rt,y=od,y error=erod] {Data/taxa-50.txt}; \label{pgfplots:label4}
        \addplot[color=red!70!black,densely dashed,mark=square*,mark options={fill=red!70!black,solid}][error bars/.cd,y dir=both,y explicit, error mark options={rotate=90,red!70!black,mark size=2pt,line width=0.3pt}] table[x=rt,y=dl,y error=erdl] {Data/taxa-50.txt}; \label{pgfplots:label5}
        \addplot[color=green!40!black,mark=*][error bars/.cd,y dir=both,y explicit, error mark options={rotate=90,densely dashed,green!40!black,mark size=2pt,line width=0.3pt}] table[x=rt,y=spr,y error=erspr] {Data/taxa-50.txt}; \label{pgfplots:label6}
        \addplot[color=white!40!black,densely dashed,mark=*,mark options={fill=white!40!black,solid}][error bars/.cd,y dir=both,y explicit, error mark options={rotate=90,densely dashed,white!40!black,mark size=2pt,line width=0.3pt}] table[x=rt,y=rf,y error=errf] {Data/taxa-50.txt}; \label{pgfplots:label7}
        \end{axis}
        &
        \begin{axis}[
            width=5cm,height=4cm,xtick={1,2,3,4},xticklabels={\emph{none},\emph{lgt},\emph{dl},\emph{both}},tick pos=left,ytick align=outside,xtick align=outside,minor y tick num=19,
            tick label style={font=\tiny},
            yticklabels={},
            ymin=0,ymax=80,title=100 taxa,enlarge y limits=false,
        ]
        \addplot[color=blue!70!black,mark=square*][error bars/.cd,y dir=both,y explicit, error mark options={rotate=90,blue!70!black,mark size=2pt,line width=0.3pt}] table[x=rt,y=od,y error=erod] {Data/taxa-100.txt};
        \addplot[color=red!70!black,densely dashed,mark=square*,mark options={fill=red!70!black,solid}][error bars/.cd,y dir=both,y explicit, error mark options={rotate=90,red!70!black,mark size=2pt,line width=0.3pt}] table[x=rt,y=dl,y error=erdl] {Data/taxa-100.txt};
        \addplot[color=green!40!black,mark=*][error bars/.cd,y dir=both,y explicit, error mark options={rotate=90,densely dashed,green!40!black,mark size=2pt,line width=0.3pt}] table[x=rt,y=spr,y error=erspr] {Data/taxa-100.txt};
        \addplot[color=white!40!black,densely dashed,mark=*,mark options={fill=white!40!black,solid}][error bars/.cd,y dir=both,y explicit, error mark options={rotate=90,densely dashed,white!40!black,mark size=2pt,line width=0.3pt}] table[x=rt,y=rf,y error=errf] {Data/taxa-100.txt};
        \end{axis}
        \\
        };
\end{tikzpicture}
\vspace*{-0.2in}
\caption{Average topological error (means with standard error bars) for species tree constructed by Only-dup, Dup-loss, SPRS, and MulRF method, for all model conditions.}
\label{fig:ATEplot}
\vspace*{-0.25in}
\end{figure}

In order to examine how Only-dup, Dup-loss, and SPRS methods perform when the process of gene tree evolution only includes events that these methods assume to be the source of discordance, we simulated gene trees that using a model that includes only duplication and loss, or LGT. While SPRS could not be tested on the former, Only-dup and Dup-loss had high ATE rate  (indicating low accuracy) on the latter.

\section{Conclusion}

We presented a new approach for inferring species tree from incongruent gene trees that is not based on potentially restrictive assumptions about the causes of the conflict among gene trees. This approach is appealing for real, genomic data sets, in which many processes such as deep coalescence, recombination, gene duplications and losses, and LGT, as well as phylogenetic error likely contribute to gene tree dischord. In simulation experiments, the MulRF method estimated species trees more accurately than other GTP methods, and it appears to be relatively robust to the effects of phylogenetic error, gene duplication and loss, and LGT.  In addition, the MulRF method is fast, estimating 100-taxon species trees from hundreds of gene trees in under two minutes. One reason for this strong performance may be the underlying unrooted metric. The advantage of an unrooted metric compared to a rooted one, like those used in the other supertree methods, has been well-studied in the context of RF supertrees for singly-labeled trees \cite{Ruchi:2012:URF}. Further tests are needed to characterize the performance of MulRF methods under different evolutionary scenarios. Another future direction will be to incorporate estimates of gene tree uncertainty into the supertree analysis by weighing the splits differently when computing the RF distance.

\bibliographystyle{abbrv}

\begin{thebibliography}{10}

\bibitem{Allen:2001:STO}
B.~L. Allen and M.~Steel.
\newblock Subtree transfer operations and their induced metrics on evolutionary
  trees.
\newblock {\em Annals of Combinatorics}, 5:1--15, 2001.

\bibitem{Ane07012007}
C.~An\'e, B.~Larget, D.~A. Baum, S.~D. Smith, and A.~Rokas.
\newblock {Bayesian} estimation of concordance among gene trees.
\newblock {\em Mol. Biol. Evol.}, 24(7):1575, 2007.

\bibitem{ArvestadBLS03}
L.~Arvestad, A.-C. Berglund, J.~Lagergren, and B.~Sennblad.
\newblock Bayesian gene/species tree reconciliation and orthology analysis
  using mcmc.
\newblock In {\em ISMB (Supplement of Bioinformatics)}, pages 7--15, 2003.

\bibitem{Avise:1983:MBE}
J.~Avise, J.~Shapira, S.~Daniel, C.~Aquadro, and R.~Lansman.
\newblock Mitochondrial {DNA} differentiation during the speciation process in
  peromyscus.
\newblock {\em Molecular Biology and Evolution}, 1:38--56, 1983.

\bibitem{Bansal:APBC:2010}
M.~S. Bansal, J.~G. Burleigh, and O.~Eulenstein.
\newblock Efficient genome-scale phylogenetic analysis under the
  duplication-loss and deep coalescence cost models.
\newblock {\em BMC Bioinformatics}, 11(Suppl 1):S42, 2010.

\bibitem{Mukul:2010:RFS}
M.~S. Bansal, J.~G. Burleigh, O.~Eulenstein, and D.~Fern{\'a}ndez-Baca.
\newblock {Robinson-Foulds} supertrees.
\newblock {\em Algorithms for Molecular Biology}, 5:18, 2010.

\bibitem{Bender:2000:TLP}
M.~A. Bender and M.~Farach-Colton.
\newblock The {LCA} problem revisited.
\newblock In G.~H. Gonnet, D.~Panario, and A.~Viola, editors, {\em LATIN},
  volume 1776 of {\em Lecture Notes in Computer Science}, pages 88--94.
  Springer, 2000.

\bibitem{Burleigh:SB:2011}
J.~G. Burleigh, M.~S. Bansal, O.~Eulenstein, S.~Hartmann, A.~Wehe, and T.~J.
  Vision.
\newblock Genome-scale phylogenetics: inferring the plant tree of life from
  18,896 discordant gene trees.
\newblock {\em Systematic Biology}, 60(2):117--125, 2011.

\bibitem{Ruchi:2012:URF}
R.~Chaudhary, J.~G. Burleigh, and D.~Fern{\'a}ndez-Baca.
\newblock Fast local search for unrooted robinson-foulds supertrees.
\newblock {\em {IEEE/ACM} Transactions on Computational Biology and
  Bioinformatics}, 9:1004--1013, 2012.

\bibitem{DegnanRosenberg2006}
J.~H. Degnan and N.~A. Rosenberg.
\newblock Discordance of species trees with their most likely gene trees.
\newblock {\em PLoS Genet}, 2(5):e68, 05 2006.

\bibitem{Doyle:1992:SB}
J.~Doyle.
\newblock Gene trees and species trees: Molecular systematics as one-character
  taxonomy.
\newblock {\em Systematic Botany}, 17:144--163, 1993.

\bibitem{Ganapathy_PhDThesis06}
G.~Ganapathy.
\newblock {\em Algorithms and Heuristics for Combinatorial Optimization in
  Phylogeny}.
\newblock PhD thesis, University of Texas at Austin, 2006.

\bibitem{Ganapathy:2006:PIB}
G.~Ganapathy, B.~Goodson, R.~Jansen, H.~Le, V.~Ramachandran, and T.~Warnow.
\newblock Pattern identification in biogeography.
\newblock {\em IEEE/ACM Transactions on Computational Biology and
  Bioinformatics}, 3:334--346, 2006.

\bibitem{Garey:1979:CIG}
M.~R. Garey and D.~S. Johnson.
\newblock {\em {Computers and Intractability: A guide to the theory of
  NP-completeness}}.
\newblock W. H. Freeman, New York, 1979.

\bibitem{Goodman:1979:FTG}
M.~Goodman, J.~Czelusniak, G.~W. Moore, A.~E. {Romero-Herrera}, and G.~Matsuda.
\newblock Fitting the gene lineage into its species lineage. a parsimony
  strategy illustrated by cladograms constructed from globin sequences.
\newblock {\em Systematic Zoology}, 28:132--163, 1979.

\bibitem{Hickey:2008:SPR}
G.~Hickey, F.~Dehne, A.~Rau-Chaplin, and C.~Blouin.
\newblock {SPR} distance computation for unrooted trees.
\newblock {\em Evolutionary Bioinformatics}, 4:17--27, 2008.

\bibitem{Huang:2009:WDA}
H.~Huang and L.~L. Knowles.
\newblock What is the danger of the anomaly zone for empirical phylogenetics?
\newblock {\em Systematic Biology}, 58:527--536, 2009.

\bibitem{Kubatko04012009}
L.~S. Kubatko, B.~C. Carstens, and L.~L. Knowles.
\newblock {STEM: species tree estimation using maximum likelihood for gene
  trees under coalescence}.
\newblock {\em Bioinformatics}, 25(7):971--973, 2009.

\bibitem{Liu06012007}
L.~Liu and D.~K. Pearl.
\newblock Species trees from gene trees: Reconstructing {Bayesian} posterior
  distributions of a species phylogeny using estimated gene tree distributions.
\newblock {\em Systematic Biology}, 56(3):504--514, 2007.

\bibitem{Maddison:1996}
W.~Maddison.
\newblock {\em {Molecular Zoology: Advances, Strategies and Protocols}},
  chapter Molecular approaches and the growth of phylogenetic biology, pages
  47--63.
\newblock Wiley-Liss,New York, 1996.

\bibitem{maddison97}
W.~P. Maddison.
\newblock Gene trees in species trees.
\newblock {\em Systematic Biology}, 46:523--536, 1997.

\bibitem{Maddison:Mesquite:2009}
W.~P. Maddison and D.~Maddison.
\newblock Mesquite: a modular system for evolutionary analysis. version 2.6.
  http://mesquiteproject.org, 2009.

\bibitem{McMorris:Steel:93}
F.~R. McMorris and M.~A. Steel.
\newblock The complexity of the median procedure for binary trees.
\newblock In {\em In Proceedings of the International Federation of
  Classification Societies}, 1993.

\bibitem{Page:1998:GCG}
R.~D.~M. Page.
\newblock {GeneTree}: comparing gene and species phylogenies using reconciled
  trees.
\newblock {\em Bioinformatics}, 14(9):819--820, 1998.

\bibitem{Pamilo:1988:MBE}
P.~Pamilo and M.~Nei.
\newblock Relationships between gene trees and species trees.
\newblock {\em Mol. Biol. Evol.}, 5:568--583, 1988.

\bibitem{Rambaut:1997:SGA}
A.~Rambaut and N.~C. Grassly.
\newblock {Seq-Gen}: An application for the {Monte-Carlo} simulation of {DNA}
  sequence evolution along phylogenetic trees.
\newblock {\em Copmput. Appl Biosci.}, 13:235--238, 1997.

\bibitem{Robinson:1981:CPT}
D.~F. Robinson and L.~R. Foulds.
\newblock Comparison of phylogenetic trees.
\newblock {\em Mathematical Biosciences}, 53:131--147, 1981.

\bibitem{Sanderson2007}
M.~J. Sanderson and M.~M. McMahon.
\newblock Inferring angiosperm phylogeny from {EST} data with widespread gene
  duplication.
\newblock {\em BMC Evolutionary Biology}, 7(suppl 1:S3), 2007.

\bibitem{Semple:2003:phy}
C.~Semple and M.~Steel.
\newblock {\em Phylogenetics}.
\newblock Oxford University Press, 2003.

\bibitem{StamatakisBioinf2006}
A.~Stamatakis.
\newblock {RAxML-VI-HPC}: Maximum likelihood- based phylogenetic analyses with
  thousands of taxa and mixed models.
\newblock {\em Bioinformatics}, 22:2688---2690, 2006.

\bibitem{SteelRodrigo08}
M.~Steel and A.~Rodrigo.
\newblock Maximum likelihood supertrees.
\newblock {\em Systematic Biology}, 57(2), April 2008.

\bibitem{SwensonBWL2010}
M.~S. Swenson, F.~Barban\c{c}on, T.~Warnow, and C.~R. Linder.
\newblock A simulation study comparing supertree and combined analysis methods
  using {SMIDGen}.
\newblock {\em Algorithms for Molecular Biology}, 5:8, 2010.

\bibitem{Swofford:1996:PI}
D.~L. Swofford, G.~J. Olsen, P.~J. Waddel, and D.~M. Hillis.
\newblock Phylogenetic inference.
\newblock In D.~M. Hillis, C.~Moritz, and B.~K. Mable, editors, {\em Molecular
  Systematics}, chapter~11, pages 407--509. Sinauer Assoc., Sunderland, Mass,
  1996.

\bibitem{Than:2011:MDC}
C.~Than and N.~Rosenberg.
\newblock Consistency properties of species tree inference by minimizing deep
  coalescences.
\newblock {\em Journal of Computational Biology}, 18:1--15, 2011.

\bibitem{Wehe:bioinfo:2008}
A.~Wehe, M.~S. Bansal, J.~G. Burleigh, and O.~Eulenstein.
\newblock Duptree: a program for large-scale phylogenetic analyses using gene
  tree parsimony.
\newblock {\em Bioinformatics}, 24(13), 2008.

\bibitem{Whidden:SPR:2012}
C.~Whidden, N.~Zeh, and R.~Beiko.
\newblock {SPRSupertrees}. version 1.1.0.
  http://kiwi.cs.dal.ca/software/sprsupertrees, 2012.

\bibitem{Yu:recomb:11}
Y.~Yu, T.~Warnow, and L.~Nakhleh.
\newblock Algorithms for {MDC}-based multi-locus phylogeny inference.
\newblock In {\em RECOMB}, pages 531--545, 2011.

\end{thebibliography}

\section*{Appendix}

\subsection*{Computing RF Distance between two mul-trees is NP-Complete}
The proof relies on a reduction from the following NP-complete \cite{Garey:1979:CIG} problem.

\begin{problem} [Exact Cover by 3-Sets (X3C)]  \\
\textit{Input:} $S := \{s_1,...,s_n\}$, where $n = 3q$, and $C := \{C_1,...,C_m\}$ such that $C_i = \{s_{i_1}, s_{i_2}, s_{i_3}\}$. \\
\textit{Output:} Are there exist sets $C_{i_1},...,C_{i_q}$ such that $\bigcup_{j=1}^{q} C_{i_j} = S$ ?  \end{problem}

Note that X3C remains NP-complete \cite{Hickey:2008:SPR} even when each element of $S$ occurs in \emph{exactly} three subsets in $C$, thus $m=n=3q$. We take this version of X3C for reduction.

Given an instance for the X3C problem, we construct two mul-trees $\mathcal{T}_1$ and $\mathcal{T}_2$ such that transforming from $\mathcal{T}_1$ into $\mathcal{T}_2$ (or vice versa) requires $\kappa$ (to be specified later) contractions and refinements if and only if an exact cover of $S$ exists.  The construction is as follows. For each $s_i \in S$, we construct two rooted binary trees $T$ and $T'$ that take a ``large'' number of contractions and refinements to transform into each other. Let $k$ and $t$ be two positive integers such that $k + 2 \geq n^2$ and $k +2 = 2^t$. Tree $T$ and $T'$ have $k + 2$ leaves. Tree $T'$ has the same topology as $T$, but for each cherry\footnote{Two leaves connected with the same internal vertex in a tree are called a \emph{cherry}.} $(x,y)$ in $T$, $x$ and $y$ are in different subtrees $T'_u$ and $T'_v$ in $T'$, where $u$ and $v$ are two children of $rt(T')$. For each $s_i \in S$, corresponding trees $T$ and $T'$ have unique leaves (see Fig.~\ref{trees}.)
\begin{figure}
\vspace*{-0.2in}
 \centering
 \includegraphics[width=2.8in]{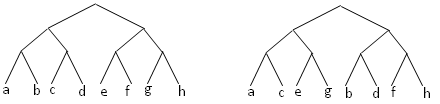}
 \vspace*{-0.15in}
 \caption{Two possible trees $T$ and $T'$ on 8 leaves with RF distance 12.}
 \label{trees}
\vspace*{-0.3in}
 \end{figure}

\begin{lemma} \label{twoT} $RF(T,T') = 2k$.\end{lemma}
\begin{proof} $RF(T,T') = 2|\mathcal{H}(T) \backslash \mathcal{H}(T')|$, since $T$ and $T'$ are binary trees. $T$ and $T'$ are binary trees on $k+2$ leaves, thus $\mathcal{H}(T) = \mathcal{H}(T') = k$. Thus it suffices to show that no cluster in $T$ matches any cluster in $T'$. Let $v \in I(T)$, the corresponding cluster $C(v)$ contains leaves of $1 \leq p \leq (k+2)/4$ cherries. From the construction, $T'$ has both leaves of each cherry in different subtrees under the root $rt(T')$; thus there is no matching cluster for $C(v)$ in $T'$. \qed \end{proof}

\begin{figure}[ht]
 \centering
 \subfigure[]{
 \includegraphics[width=2.5in]{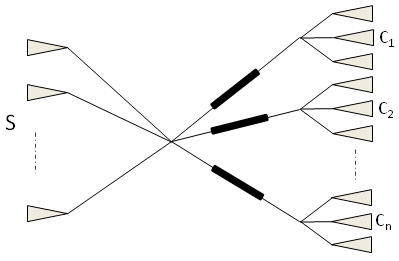}
 \label{fig:main}
 }
 \subfigure[]{
 \includegraphics[width=2in]{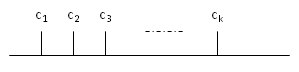}
 \label{fig:toll}
 }

 \label{fig:struct}
 \vspace*{-0.15in}
 \caption{(a) Structure of mul-tree $\mathcal{T}_1$ and (b) A toll sequence of $k$ leaves.}
 \vspace*{-0.25in}
 \end{figure}

 We are now ready for the construction of $\mathcal{T}_1$ and $\mathcal{T}_2$. 
Figure \ref{fig:main} outlines the structure of $\mathcal{T}_1$. The solid rectangles represent \emph{toll} sequences of $k$ uniquely labeled leaves (Fig. \ref{fig:toll}). The left side of $\mathcal{T}_1$ has $n$ triangles one for each of the $n$ elements in $S$. Each triangle represents a tree $T$ corresponding to $s_i \in S$, connecting through its root. The right side of $\mathcal{T}_1$ has $n$ sets of 3 triangles corresponding to the subsets in $C$; for each subset $C_i = \{s_{i_1}, s_{i_2}, s_{i_3}\}$, the triangles represent three trees $T'$s, corresponding to each $s_{i_j}$ (for $1 \leq j \leq 3$), connected through their roots. 

$\mathcal{T}_2$ has the similar structure except that $\mathcal{T}_2$ has tree $T'$ for each $s_i \in S$ and tree $T$ for each element of $C_i \in C$ (for $1 \leq i \leq n$). Thus, $\mathcal{T}_2$ has $T'$s on the left side and $T$s on the right side, which is opposite to what $\mathcal{T}_1$ has.

\begin{lemma} Mul-trees $\mathcal{T}_1$ and $\mathcal{T}_2$ can be constructed in polynomial time. \end{lemma}
\begin{proof} Trees $T$ and $T'$ are rooted binary trees on $k + 2$ leaves. $T$ and $T'$ can be constructed in polynomial time, and so the $4n$ copies of each (for $\mathcal{T}_1$ and $\mathcal{T}_2$). Further, $2n$ toll sequences ($n$ for each $\mathcal{T}_1$ and $\mathcal{T}_2$) can be constructed in polynomial time. There are constant number of rest of the vertices in $\mathcal{T}_1$ and $\mathcal{T}_2$. Hence, the Lemma. \qed \end{proof}

Here is the connection between exactly covering $S$ and transforming $\mathcal{T}_1$ into $\mathcal{T}_2$ by contractions and refinements: To transform $\mathcal{T}_1$ into $\mathcal{T}_2$, all we need is to convert each tree $T$ on the left into $T'$ and each tree $T'$ on the right into $T$. From Lemma \ref{twoT}, this costs $24qk$ contractions and refinements. A rather clever technique is to swap $3q$ $T$s on the left with their counterparts on the right and to transform the remaining $6q$ $T'$s on the right into $T$s. If an exact cover $C_{i_1},...,C_{i_q}$ of $S$ exists, we can partition the $3q$ $T$s into $q$ groups according to the cover. For each $C_j$ ($j = i_1,...,i_q$) in the cover, we swap the corresponding group of trees for sequences $s_{j_1}, s_{j_2}, s_{j_3}$ with their counterparts.

\begin{lemma} \label{swapT} All $T'$s for each $C_j$ ($j = i_1,...,i_q$) can be swapped with corresponding $T$s by $2(k+1)$ contractions and refinements. \end{lemma}
\begin{proof} Take the toll sequence corresponding to $C_j$ and contract its $k+1$ edges; i.e., $(k-1)$ internal edges and 2 edges at both the sides of the toll sequence. Now refine it so that corresponding $T$s move in $C_j$ and $T'$s stay in the left. This takes $2(k+1)$ contractions and refinements. \qed \end{proof}

From Lemma \ref{swapT}, if the exact cover of S exists, then $6q$ trees can be transformed by $2q(k+1)$ contractions and refinements. Remaining $6q$ $T'$s can be transformed into $T$s by $12qk$ contractions and refinements. Hence, we have the following lemma.

\begin{lemma} If set $S$ has an exact cover then the RF distance between $\mathcal{T}_1$ and $\mathcal{T}_2$ is $\kappa = 2q(k+1) + 12kq$. \end{lemma}

If there is no exact cover of $S$, then either more than $6q$ trees ($T$ or $T'$) are transformed separately or more than $q$ group swaps are performed. The construction guarantees that both cases will cost more than the cost of transforming ($\mathcal{T}_1$ into $\mathcal{T}_2$) in exact cover case. Hence, we conclude the following.

\begin{theorem} Set $S$ has no exact cover if and only if the RF distance between $\mathcal{T}_1$ and $\mathcal{T}_2$ is more than $\kappa = 2q(k+1) + 12kq$. \end{theorem}

\subsection*{Other Proofs}

\begin{proof}[Theorem \ref{thm:diffr}] Let the given input mul-tree $\mathcal{T}$ is such that $\mathcal{T} := (T,M,\varphi)$. We prove the Theorem by showing that for each $a \in M$, where $|\varphi^{-1}(a)| = k$, all $k!$ ways of uniquely relabeling corresponding $k$ leaves in both $\mathcal{T}$ and $\mathcal{S}$ result into the same number of matched and unmatched splits in the corresponding mutually consistent full differentiations. The set of splits in $\mathcal{T}$ can be divided into two categories:
\begin{itemize}
  \item \emph{Category 1:} Splits that have all the leaves labeled with $a$ in one part. Such a split will always have a match irrespective of the labeling.
  \item \emph{Category 2:} The remaining splits. Such splits are not present in $\mathcal{S}$, therefore, they will never have a match irrespective of the labeling.
  \qed
\end{itemize}
\end{proof}

\begin{proof}[Observation 1] Let the extension of $S$ be $\mathcal{S} := (T',M',\varphi')$. Let $\mathbf{S}$ be a full differentiation of $\mathcal{S}$ that is consistent with $\mathbf{T}$, where $T'$ and $\mathbf{S}$ are isomorphic under bijection $\tau : V(T') \rightarrow V(\mathbf{S})$.

For $Z \in \{X,Y\}$, let $\mathbf{S}[Z] = \{l \in \mathcal{L}(\mathbf{S}): \varphi'(\tau^{-1}(l)) \in \mathcal{L}(Z)\}$.

Since, $\mathcal{L}(\mathbf{S}_{|\mathcal{L}(\mathbf{T})}) \cap \mathbf{S}[Z] = \emptyset$,  $RF(\mathbf{S}_{|\mathcal{L}(\mathbf{T})},\mathbf{T}) =$ $RF(\mathbf{S}'_{|\mathcal{L}(\mathbf{T})},\mathbf{T})$. Now, $RF(\mathcal{S},\mathcal{T}) =$ $RF(\mathbf{S},\mathbf{T}) =$ $RF(\mathbf{S}_{|\mathcal{L}(\mathbf{T})},\mathbf{T}) =$ $RF(\mathbf{S}'_{|\mathcal{L}(\mathbf{T})},\mathbf{T}) =$ $RF(\mathbf{S}',\mathbf{T})$ = $RF(\mathcal{S}',$ $\mathcal{T})$.  \qed
\end{proof}

\begin{proof}[Lemma \ref{lm:NNI-RF}]
\begin{align}
RF(\mathbb{S}'',\mathbb{T}) &= |\mathcal{L}(\mathbb{T})|-|I(\mathbb{T})| -2 + 2|\mathcal{F}_{\mathbb{S}''}| \nonumber\\
&=|\mathcal{L}(\mathbb{T})|-|I(\mathbb{T})| -2 \nonumber\\
  & \hspace*{4mm} + 2|\{u \in I(\mathbb{T}): f_{\mathbb{S}''}(u) = 0\}| \nonumber\\
&=|\mathcal{L}(\mathbb{T})|-|I(\mathbb{T})| -2 + 2|\mathcal{F}_{\mathbb{S}'}| \nonumber\\
 & \hspace*{4mm} -2|\{u \in H: f_{\mathbb{S}'}(u) = 0 \hspace*{1mm} \& \hspace*{1mm} f_{\mathbb{S}''}(u) \geq 1 \}| \nonumber\\
 & \hspace*{4mm} + 2|\{u \in H: f_{\mathbb{S}''}(u) = 0 \hspace*{1mm} \& \hspace*{1mm} f_{\mathbb{S}'}(u) \geq 1 \}| \nonumber\\
&=RF(\mathbb{S}',\mathbb{T}) - 2|G| + 2|L|  \nonumber
\end{align}   \qed
\end{proof}

\begin{proof}[Lemma \ref{lm:comp}] The RF distance computation for $\overline{S}$, obtained by pruning $Y$ and regrafting at a leaf in $X$, can be done in $\Theta(n)$ time. After $\overline{S}$, the RF distance for each tree $S'$, obtained by regrafting $Y$ on each edge in $X$, can be computed in constant time by performing regrafts in the order of $\aleph$. There are $\Theta(n)$ edges in $\aleph$, thus the RF computation for all the trees can be done in $\Theta(n)$ time. The same argument applies for pruning $X$ and regrafting on the edges in $Y$.  \qed \end{proof}

\begin{proof}[Theorem \ref{tm:main}] There are $\Theta(n)$ internal edges in $S$. For each edge $\{x,y\}$ in $S$, where $X$, $Y$ be two resulting subtrees containing $x$, $y$, respectively. The RF distance for all the trees obtained by regrafting $X$ (or $Y$) on each edge in $Y$ (or $X$) can be computed in $\Theta(n)$ time from Lemma \ref{lm:comp}. Thus for $k$ input trees the RF distance can be checked in $\Theta(nk)$ time. The total time over all $\Theta(n)$ internal edges is $\Theta(n^2k)$. \qed \end{proof}

\end{document}